\documentclass[12pt]{article}
\usepackage{mathrsfs}
\usepackage{dsfont}
\usepackage{amsthm}
\usepackage{mathrsfs}
\usepackage{amsmath}
\usepackage{amsfonts}
\usepackage[colorlinks,linkcolor=black,anchorcolor=blue,citecolor=blue]{hyperref}
\usepackage{amssymb, amsmath, cite}

\newtheorem{theorem}{Theorem}[section]
\newtheorem{lemma}{Lemma}[section]

\newcommand\be{\begin{equation}}
\newcommand\ee{\end{equation}}
\newcommand\ber{\begin{eqnarray}}
\newcommand\eer{\end{eqnarray}}
\newcommand\berr{\begin{eqnarray*}}
\newcommand\eerr{\end{eqnarray*}}
\newcommand\bea{\begin{eqnarray}}
\newcommand\eea{\end{eqnarray}}

\newcommand{\nn}{\nonumber}
\newcommand\vep{\varepsilon}\newcommand{\ii}{\mathrm{i}}
\newcommand{\dd}{\mathrm{d}}
\newcommand\e{\mathrm{e}}\newcommand\pa{\partial}
\newcommand\var{\vartheta}
{\setlength\arraycolsep{2pt}

\begin{document}

\title{Twisted vortices in two--component Ginzburg--Landau theory}
\author{Lei Cao, Shouxin Chen\\School of Mathematics and Statistics, Henan University\\
Kaifeng, Henan 475004, PR China }
\date{}
\maketitle

\begin{abstract}

In this note, a brief introduction to the physical and mathematical background of the two--component Ginzburg--Landau theory is given. From this theory we derive a boundary value problem whose solution can be obtained in part by solving a minimization problem using the technique of variational method, except that two of the eight boundary conditions cannot be satisfied. To overcome the difficulty of recovering the full set of boundary conditions, we employ a variety of methods, including the uniform estimation method and the bounded monotonic theorem, which may be applied to other complicated vortex problems in gauge field theories. The twisted vortex solutions are obtained as energy--minimizing cylindrically symmetric field configurations. We also give the sharp asymptotic estimates for the twisted vortex solutions at the origin and infinity.

\end{abstract}

\medskip
\begin{enumerate}

\item[]

\item[]
{MSC numbers(2020):} 35J50, 81T10.

\end{enumerate}

\section{Introduction}\label{s0}
\setcounter{equation}{0}

A twisted vortex is a helical configuration of vortices that encircles a common axis. This phenomenon has been observed in superconducting current-carrying wires in the presence of a parallel external magnetic field, as documented in references \cite{Wa,Br,Ka}. Additionally, spontaneous emergence of twisted vortices has been reported in vortex lines expanding into vortex--free rotating superfluid, as corroborated by numerical calculations \cite{El}.

Ginzburg--Landau theory serves as a fundamental framework for understanding phase transitions \cite{Sc}, critical phenomena \cite{Mi,De}, and non--equilibrium patterns \cite{Ho} in systems ranging from superconductors \cite{Tr} to magnets \cite{Ros}. Multi--component Ginzburg--Landau theory \cite{As,Be,Le,Xu,Zh,Pe,Y,Ye,Fe} showcases its versatility and applicability in various physical phenomena, including superconductivity and vortex dynamics. In recent research, two--component Ginzburg--Landau theories \cite{Du,Ki,Ma} have been a topic of interest for the study of symmetry near the vortex core and are crucial for studying unconventional superconductors. Besides, research on two--component Ginzburg--Landau theory has provided valuable insights into other complex systems through the study of the conditions of validity, the applicability of the model, and the implications for different types of order parameters. In addition, when an external magnetic field is applied, two--component Ginzburg--Landau theories have been successfully used to describe type I.5 superconductors \cite{B,Ba,Ca,Si}, where the penetration depth of the magnetic field lies between the coherence lengths of the different order parameters. The microscopic derivation of two--component Ginzburg--Landau models and the application of variational methods further demonstrate the utility and versatility of this theory in investigating complex physical phenomena. Two--component Ginzburg--Landau models have also been applied to systems such as the $MgB_2$ materials \cite{Ro}, where a variational method is used to study superconductors. However, Babaev et al. \cite{Bab} question the validity of widely used two--component Ginzburg--Landau models, arguing against the claim that there should be no disparity in the coherence lengths of two superconducting components described by a Ginzburg--Landau theory. Furthermore, the Ginzburg--Landau theory has been extended to study skyrmions in inversion--symmetric magnets with competing interactions \cite{Li}. This general Ginzburg--Landau theory for skyrmions is valid in the long--wavelength limit, demonstrating the versatility of the Ginzburg--Landau formalism in describing a wide range of physical systems.

The present work aims to investigate vortices in the multi--component Ginzburg--Landau theory, which may be referred to as the Abelian gauge model with an extended scalar sector. In particular, we will concentrate on the two--component Ginzburg--Landau theory, which has the most general $U(1)\times U(1)$ symmetric scalar potential proposed by Witten \cite{Wi}. In the case of the potential minimum, where at least one of the fields assumes a non--zero value, the requisite details have not been provided thus far.

The remainder of this paper is organized as follows. Section 2 gives a brief review of the two--component Ginzburg--Landau theory, derives the field equations and the associated minimization problem, and states the main result. In Section 3, we demonstrate the existence of a solution to the energy--minimizing problem and verify the boundary conditions of the field equations. In Section 4, the decay estimate properties of the solution are established.

\section{Twisted vortices model and existence result}\label{s1}
\setcounter{equation}{0}

The Lagrangian of the two--component Ginzburg--Landau theory is
\be\label{1.1}
\mathcal{L}=\frac{1}{e^2}\left\{-\frac{1}{4}F_{\mu\nu}F^{\mu\nu}+(D_\mu\Phi)^\dag(D_\mu\Phi)-V(\Phi,\Phi^\dag)\right\},
\ee
where $F_{\mu\nu}=\pa_\mu A_\nu-\pa_\nu A_\mu$, $D_\mu\phi_a=\partial_\mu-\ii e_a A_\mu$ is the standard gauge covariant derivative. For later use, we assume general couplings, $(e_1,e_2)$, of $\Phi=(\phi_1,\phi_2)^T$ to the $U(1)$ gauge field, and $V$ is the most general $U(1)\times U(1)$ symmetric self--interaction potential
\be\label{1.2}
V=\frac{\beta_1}{2}(|\phi_1|^2-1)^2+\frac{\beta_2}{2}|\phi_2|^4+\beta'|\phi_1|^2|\phi_2|^2-\alpha|\phi_2|^2,
\ee
which is given by Witten \cite{Wi} and contains four real parameters, $\beta_1, \beta_2, \beta'$ and $\alpha$. One of the $U(1)$ gauge symmetry acts on the field as $\Phi\rightarrow\exp(\ii \chi)\Phi$, $A_\mu\rightarrow A_\mu+\pa_\mu \chi$, where $\chi=\chi(x)$ is the gauge function. The other $U(1)$ symmetry is global, and it acts on the fields as $\phi_1\rightarrow\exp(-\ii \alpha)\phi_1$, $\phi_2\rightarrow\exp(-\ii \alpha)\phi_2$, where $\alpha$ is a constant.

The field equations derived from the Lagrangian \eqref{1.1} are
\bea
\pa^\rho F_{\rho\mu}&=&\ii \sum_a e_a\{(D_\mu\phi_a)^\ast\phi_a-\phi_a^\ast D_\mu\phi_a\},\label{1.3}\\
D_\rho D^\rho \Phi&=&-\pa V(\Phi^\dag, \Phi)/\pa \Phi^\dag.\label{1.4}
\eea

By rotating the phases of $\phi_a$ separately, it can be shown that the conserved currents that correspond to the two $U(1)$ symmetries of the theory \eqref{1.1} have the following form
\be\label{1.5}
j_\mu^{(a)}=-\ii(\phi_a^\ast D_\mu\phi_a-\phi_a(D_\mu\phi_a)^\ast).
\ee
The electrical current in \eqref{1.3}--\eqref{1.4} is given by $j_\mu=\sum_a e_a j_\mu^{(a)}$. Furthermore, the additional symmetry can be obtained by changing the phase difference
\be\label{1.6}
j_\mu^3=j_\mu^{(1)}-j_\mu^{(2)}=-\ii(\phi_1^\ast D_\mu\phi_1-\phi_2^\ast D_\mu\phi_2-\phi_1(D_\mu\phi_1)^\ast+\phi_2(D_\mu\phi_2)^\ast),
\ee
which is consistent with the third isospin component of the global $SU(2)$ current of the semilocal theory \cite{F,Fo}.

A suitably reduced stationary, cylindrically symmetric ansatz in the radial gauge can be written as
\bea
\phi_1(r,\var,z)&=&f(r)\e^{\ii N\var},~~~~~~~A_\var(r,\var,z)=Na(r),\label{1.7}\\
\phi_2(r,\var,z)&=&g(r)\e^{\ii M\var}\e^{\ii \omega z},~~A_3(r,\var,z)=\omega b(r),\label{1.8}
\eea
with $A_0=A_r=0$ and $\omega$ is the real twist parameter. By applying the ansatz \eqref{1.7}--\eqref{1.8} to the field equations \eqref{1.3}--\eqref{1.4}, we obtain the reductions
\bea
&&r\left(\frac{a'}{r}\right)'=2f^2e_1(e_1a-1)+2g^2e_2(e_2a-M/N),\label{1.9}\\
&&\frac{1}{r}(r b')'=2b(e_1^2f^2+e_2^2g^2)-2e_2g^2,\label{1.10}\\
&&\frac{1}{r}(r f')'=f\left(\frac{(1-e_1a)^2N^2}{r^2}+e_1^2\omega^2b^2+\beta_1(f^2-1)+\beta'g^2\right)\label{1.11}\\
&&\frac{1}{r}(r g')'=g\left(\frac{(e_2Na-M)^2}{r^2}+\omega^2(1-e_2b)^2+\beta_2g^2-\alpha+\beta'f^2\right).\label{1.12}
\eea
The energy density for \eqref{1.9}--\eqref{1.12} is given by
\bea\label{1.13}
\mathcal{E}=&&\frac{1}{2}\left(\frac{N^2(a')^2}{r^2}+\omega^2(b')^2\right)+(f')^2+(g')^2+\frac{N^2(1-e_1a)^2}{r^2}f^2+\frac{(e_2Na-M)^2}{r^2}g^2\nn\\
&&+\omega^2\left(e_1^2b^2f^2+(1-e_2b)^2g^2\right)+V(f,g),
\eea
with $V(f,g)=\beta_1(f^2-1)^2/2+\beta_2g^4/2+\beta'f^2g^2-\alpha g^2$. Thus, for per unit length, the total energy associated with \eqref{1.9}--\eqref{1.12} is found to be as the integral over the plane of $\mathcal{E}$,
\be\label{1.14}
E=2\pi\int_0^\infty r\dd r \mathcal{E}.
\ee

We need to specify boundary conditions for equations \eqref{1.9}--\eqref{1.12}. First, we see from \eqref{1.7}--\eqref{1.8} and the regularity requirement that $f$ and $g$ must satisfy
\be\label{1.15}
\lim_{r\to 0}\left(f(r), g(r)\right)=(0, 0).
\ee
Secondly, finite energy condition
\be\label{1.16}
E<\infty
\ee
implies that $f(r)\to 1, g(r)\to 0, a(r)\to 1/e_1$ and $b(r)\to 0$ as $r\to\infty$. Besides, it can be seen from \eqref{1.16} that $a(r)\to$ some constant $C_0$ and $b(r)\to$ some constant $C_1$ as $r\to 0$. However, $C_0$ and $C_1$ cannot be determined complete. For convenient, we may assume $C_0=0$, then these boundary conditions can be recorded as follows
\bea
a(0)&=&0,~~~b(0)=C_1,~~~f(0)=0,~~~g(0)=0,\label{1.17}\\
a(\infty)&=&1/e_1,~~b(\infty)=0,~~f(\infty)=1,~~g(\infty)=0.\label{1.18}
\eea

Our main result on twisted vortices in the extended Abelian Higgs model, which is governed by the boundary value problem consisting of \eqref{1.9}--\eqref{1.12} and \eqref{1.17}--\eqref{1.18}, may be stated as follows.
\begin{theorem}\label{th2.1}
For any prescribed parameters $\beta_1, \beta_2, \beta', \alpha, M, N, e_1, e_2$ satisfying the condition
\be\label{1.19}
\beta_1,~\beta_2>0,~~\alpha\geq\beta'>0,~~\beta_1\beta_2>\beta'\alpha,~~M>N>0,
\ee
the twisted vortices equations \eqref{1.9}--\eqref{1.12} have a least energy solution $(a,b,f,g)$ which minimizes the energy \eqref{1.14} and enjoys the boundary conditions \eqref{1.17}--\eqref{1.18}. Moreover, there hold the properties $0<a(r)<1/e_1$, $b(r)$ is between $C_1$ and $0$, and $0<f(r)<1$ for all $r>0$, and obey the sharp decay estimates
\bea
a(r)&=&C_0+O(r^{2(1-\vep)}),~~b(r)=C_1+O(r^{2N(1-\vep)+2}),\nn\\
f(r)&=&O(r^{N(1-\vep)}),~~g(r)=O(r^{M(1-\vep)}),\nn
\eea
at origin, and
\bea
a(r)&=&1/e_1+O(\e^{-\sqrt{2}e_1(1-\vep)r}),~~b(r)=O(\e^{-\sqrt{2}e_1(1-\vep)r}),\nn\\
f(r)&=&O(\e^{-\sqrt{2}\beta_1(1-\vep)r}),~~g(r)=O(\e^{-\sqrt{\omega^2-\alpha+\beta'}(1-\vep)r}),\nn
\eea
at infinity, where the twist number $\omega$ satisfies $\omega^2>\alpha-\beta'$ and $\vep$ is an arbitrary number lying in the interval $(0,1)$.
\end{theorem}

\section{Solutions to equations of motion}\label{s2}
\setcounter{equation}{0}

In this section we present and comment on the mathematical details of the result given in Theorem \ref{th2.1}. Methods of energy minimization and asymptotic realization will be used.

First, we concentrate on the potential function $V(f,g)$ in \eqref{1.13}.

\begin{lemma}\label{le.1}
Suppose the parameters $\beta_1,~\beta_2,~\beta'$ and $\alpha$ satisfy
\be\label{1.19b}
\beta_1,~\beta_2>0,~\alpha\geq\beta'>0,~\beta_1\beta_2>\beta'\alpha,
\ee
 then the potential $V(f, g)$ in \eqref{1.13} has a zero bound from below.
\end{lemma}
\begin{proof}
It will be convenient to use the substitution
\be
x=f^2,~y=g^2,\nn
\ee
then the potential function can be written as
\be\label{2.2}
V(x, y)=\frac{\beta_1}{2}(x-1)^2+\frac{\beta_2}{2}y^2+\beta'xy-\alpha y.
\ee
Setting
\bea
V_x&=&\beta_1(x-1)+\beta'y=0,\nn\\
V_y&=&\beta_2y+\beta'x-\alpha=0,\nn
\eea
thus we yield the single critical point $(x_0, y_0)$ and
\be\label{2.3}
x_0=\frac{\beta_1\beta_2-\beta'\alpha}{\beta_1\beta_2-(\beta')^2},~~y_0=\frac{\beta_1(\alpha-\beta')}{\beta_1\beta_2-(\beta')^2},
\ee
which are nonnegative by virtue of \eqref{1.19b}. Besides, since
\be
V_{xx}=\beta_1>0,~~V_{xy}=\beta',~~V_{yy}=\beta_2>0,\nn
\ee
we conclude that the critical point $(x_0, y_0)$ corresponds to a local maximum and the value of $V$ at this point is
\be\label{2.4}
V(x_0, y_0)=\frac{\beta_1}{2}(\alpha-\beta')\left(\alpha(\beta')^2+(\beta_1\beta_2-(\beta')^2)\beta'\right)\geq 0.
\ee
Furthermore, $(f, g)\to (1, 0)$ as $r\to\infty$, and $V(1, 0)\equiv 0$; $(f, g)\to (0, 0)$ as $r\to 0$, and $V(0, 0)=\beta_1/2>0$.

In conclusion, we always have the potential $V\geq 0$. The proof is complete.
\end{proof}

Note that all terms in the energy functional \eqref{1.14} are nonnegative except the term $-\alpha g^2(r)$ in the potential $V$. This means that the energy functional \eqref{1.14} is bounded from below.

With preparation given above, we now are ready to solve the two--point boundary value problem \eqref{1.9}--\eqref{1.12} with \eqref{1.17}--\eqref{1.18}.

Set
\be\label{2.1}
\eta_0=\inf\{E(a, b, f, g)~|~(a, b, f, g)\in \mathcal{X}\},
\ee
where the admissible space is defined by
\bea\label{2.1b}
\mathcal{X}=&&\{(a, b, f, g)~|~E(a, b, f, g)<\infty, a, b, f, g~\text{are continuous on}~(0,\infty)~\text{and}\nn\\
&&\text{absolutely continuous in any compact subinterval of}~(0, \infty)~\text{and satisfy}\nn\\
&&a(0)=0, f(0)=0, g(0)=0, b(\infty)=0, f(\infty)=1, g(\infty)=0\}.
\eea

Note that the structure of the functional \eqref{1.14} indicates that we can always modify $(a, b, f, g)$ in $\mathcal{X}$ if necessary to lower the energy to obtain the property
\be\label{2.5}
0\leq a\leq 1/e_1,~~b~\text{is between}~C_1~\text{and}~0,~~0\leq f(r)\leq 1.
\ee
Thus, from now on, we always observe this assumption.

Let $\{(a_n, b_n, f_n, g_n)\}$ be a minimizing sequence of \eqref{2.1}. It is easy to see that for any pair of numbers $0<p<q<\infty$, $\{(a_n, b_n, f_n, g_n)\}$ is a bounded sequence in $W^{1,2}(p,q)$. By a diagonal subsequence argument, we can obtain the existence of a quartet $a, b, f, g\in W_{\text{loc}}^{1,2}(0,\infty)$, so that $E(a, b, f, g)<\infty$ and $a_n\to a, b_n\to b, f_n\to f, g_n\to g~(n\to\infty)$ weakly in $W^{1,2}(p,q)$ and strongly in $C[p, q]$ for any $0<p<q<\infty$. Once the boundary conditions are verified, all that remains is to show the weak semicontinuity of $E$ over $\mathcal{X}$.

Let
\bea\label{2.12b}
&&E(a,b,f,g;p,q)\nn\\
=&&2\pi\int_p^q\Big\{\frac{1}{2}\left(\frac{N^2(a')^2}{r^2}+\omega^2(b')^2\right)+(f')^2+(g')^2+\frac{N^2(1-e_1a)^2}{r^2}f^2+\frac{(e_2Na-M)^2}{r^2}g^2\nn\\
&&+\omega^2\left(e_1^2b^2f^2+(1-e_2b)^2g^2\right)+\beta_1(f^2-1)^2/2+\beta_2g^4/2+\beta'f^2g^2-\alpha g^2\Big\}r \dd r,
\eea
where the pair numbers $0<p<q<\infty$ and $a,b,f,g$ are absolutely continuous over $(0,\infty)$. Then, it is clear that
\be\label{2.13c}
\lim_{n\to\infty}(E(a_n,b_n,f_n,g_n;p,q)-E(a,b,f,g;p,q))=0.
\ee
Hence, we have
\be\label{2.13d}
E(a,b,f,g;p,q)\leq \liminf_{n\to\infty}E(a_n,b_n,f_n,g_n;p,q).
\ee

In view of \eqref{2.1b}, there holds
\be\label{2.13e}
\lim_{r\to 0} r(g_n^2-g^2)=0.
\ee
Besides, the boundary condition and asymptotic behavior of $g(r)$ at infinity (we will show these results in the later) tells us
\be\label{2.13f}
\lim_{r\to \infty} r(g_n^2-g^2)=0.
\ee
Combining \eqref{2.13c}, \eqref{2.13e} and \eqref{2.13f}, we deduce that
\be\label{2.13g}
\lim_{n\to\infty}\int_0^\infty r g_n^2\dd r=\int_0^\infty r g^2\dd r.
\ee
Therefore, let $p\to 0$ and $q\to\infty$ in \eqref{2.13d}, we can conclude that
\be\label{2.13h}
E(a,b,f,g)=E(a,b,f,g;0,\infty)\leq \eta_0.
\ee

Furthermore, in order to prove that $(a,b,f,g)\in \mathcal{X}$ is the solution of the minimization problem \eqref{2.1}, it is necessary to show that $(a,b,f,g)$ lies in $\mathcal{X}$. In other words, we need to verify the boundary conditions stated in the admissible space $\mathcal{X}$. We proceed as follows.

Assuming
\be\label{2.5b}
E(a_n, b_n, f_n, g_n)\leq\eta_0+1~~\text{for all}~n.
\ee
Then
\be\label{2.6}
|a_n(r)-0|\leq\int_0^r|a_n'(s)|\dd s\leq\left(\int_0^r s\dd s\right)^{\frac 12}\left(\int_0^r\frac{(a_n'(s))^2}{s}\dd s\right)^{\frac 12}\leq\frac{\eta_0+1}{\sqrt{2}\pi N^2}r
\ee
showing that $a_n(r)\to 0$ as $r\to 0$ uniformly. Besides, we observe that
\be\label{2.7}
|f_n^2(r_1)-f_n^2(r_2)|\leq 2\int_{r_1}^{r_2}|f_n'(r)f_n(r)|\dd r\leq 2\left(\int_{r_1}^{r_2}r (f_n'(r))^2 \dd r\right)^{\frac 12}\left(\int_{r_1}^{r_2}\frac{f_n^2(r)}{r} \dd r\right)^{\frac 12}\to 0
\ee
as $r_1, r_2\to 0$ due to \eqref{2.5b} and the property that $a(r)\to 0$ as $r\to 0$. So $\lim_{r\to 0}f_n(r)$ exists, says $f^*$. In fact, $f^*=0$. Otherwise, the integral
\be
\int_0^\infty\frac{N^2(1-e_1a(r))^2}{r}f^*\dd r\nn
\ee
diverges at $r=0$ which is a contradiction to the finite energy. Similarly, using $a(r)\to 0$ as $r\to 0$ and the assumption
\be\label{2.7b}
\int_0^\infty (rg'^2+\frac{(e_2Na-M)^2}{r}g^2)\dd r \leq\eta_0+1,
\ee
we see
\be\label{2.8}
|g_n^2(r_1)-g_n^2(r_2)|\leq 2\int_{r_1}^{r_2}|g_n'(r)g_n(r)|\dd r\leq 2\left(\int_{r_1}^{r_2}r (g_n'(r))^2 \dd r\right)^{\frac 12}\left(\int_{r_1}^{r_2}\frac{g_n^2(r)}{r} \dd r\right)^{\frac 12}\to 0
\ee
as $r_1, r_2\to 0$. Then \eqref{2.7b} and \eqref{2.8} lead to $\lim_{r\to 0}g_n(r)=0$. Therefore, $f_n(r)\to 0$ and $g_n(r)\to 0$ uniformly as $r\to 0$.

In order to show that $\lim_{r\to\infty}f_n(r)=1$, we set
\be\label{2.9}
F_n(r)=f_n(r)-1,
\ee
then $-1\leq F_n(r)\leq0$. In terms of \eqref{2.5b}, there exists some large $R_0>0$, such that
\bea\label{2.10}
\eta_0+1&\geq&\int_{R_0}^\infty \left(r(f_n'(r))^2+\frac{\beta_1 r}{2}(f_n^2(r)-1)^2\right)\dd r\nn\\
&\geq&\int_{R_0}^\infty \left(r(F_n'(r))^2+\frac{\beta_1 r}{2}F_n^2(r)(F_n(r)+2)^2\right)\dd r\nn\\
&\geq&\int_{R_0}^\infty \left(r(F_n'(r))^2+\frac{\beta_1 r}{2}F_n^2(r)\right)\dd r\nn\\
&\geq& R_0\int_{R_0}^\infty \left((F_n'(r))^2+\frac{\beta_1 }{2}F_n^2(r)\right)\dd r,
\eea
which means $F_n(r)\in W^{1,2}(R_0,\infty)$. Thus, we get $F_n(\infty)=0$. Consequently, $f_n(r)\to 1$ as $r\to\infty$ uniformly. This fact shows that there exists a sufficiently large constant $R_1>0$ such that $f_n(r)>\frac 12$ for $r\geq R_1$ so that
\be\label{2.11}
\eta_0+1\geq2\pi \omega^2\int_0^\infty r\left(\frac{(b_n'(r))^2}{2}+e_1^2b_n^2(r)f_n^2(r)\right)\dd r\geq2\pi \omega^2 R_1\int_{R_1}^\infty\left(\frac{(b_n'(r))^2}{2}+\frac{e_1^2b_n^2(r)}{4}\right)\dd r
\ee
when $r$ is sufficiently large in view of the finite energy \eqref{2.5b}. Therefore, $b_n(r)\in W^{1,2}(R_1,\infty)$ which implies $b_n(r)\to 0$ as $r\to\infty$ uniformly.

Furthermore, according to $b_n(\infty)=0$, we see that there exists a sufficiently large constant $R_2>0$ so that $(1-e_2 b_n(r))^2>\frac 14$ for $r>R_2$. Using the same method as showing that $b_n(\infty)=0$, we can deduce that $g_n(r)\in W^{1,2}(R_2, \infty)$. Of course, $g_n(r)\to 0$ as $r\to\infty$ uniformly.

So far we have proved that the boundary conditions given in $\mathcal{X}$ are satisfied.

Hence, $(a,b,f,g)\in \mathcal{X}$ solves the minimization problem \eqref{2.1}. While, in order to show that the minimizer obtained in \eqref{2.1} is a solution of the equations \eqref{1.9}--\eqref{1.12} under the constraints \eqref{1.17}--\eqref{1.18}, it remains to establish that $a(\infty)=1/e_1$ and $b(0)=C_1$. To this end, we will use the equations \eqref{1.9} and \eqref{1.10}.

\begin{lemma}\label{le.2}
Suppose $(a,b,f,g)$ solves the minimizer problem \eqref{2.1}. Then the functions $a$ and $b$ fulfill the desired boundary conditions
\be\label{2.16}
\lim_{r\to\infty}a(r)=\frac{1}{e_1},~~~\lim_{r\to 0}b(r)=C_1.
\ee
\end{lemma}
\begin{proof}
Recall that $(a,b,f,g)$ solves the minimizer problem \eqref{2.1}, we see that $a$ and $b$ satisfy the equations \eqref{1.9} and \eqref{1.10}.

First, we show that $\lim_{r\to\infty}a(r)=\frac{1}{e_1}$. We claim
\be\label{2.18}
\liminf_{r\to\infty}\frac{a'(r)}{r}=0.
\ee
Otherwise, there are constant $\vep_0>0$ and $R_3>0$ such that
\be\label{2.19}
\left|\frac{a'(r)}{r}\right|>\vep_0~~\text{for}~~r>R_3.
\ee
Then, using \eqref{2.19}, we arrive at
\be\label{2.20}
\int_r^{\infty}\frac{(a'(s))^2}{s}\dd s>\int_r^{\infty}s\vep_0^2 \dd s=\infty,
\ee
which contradicts the convergence of the integral $\int_0^\infty\frac{(a'(r))^2}{r}\dd r$.

Since $f(\infty)=1$ and $g(\infty)=0$, we get from \eqref{1.9} that the inequality
\be\label{2.17}
\left(\frac{a'}{r}\right)'<\frac{e_1^2}{r}\left(a-\frac{1}{e_1}\right)
\ee
holds when $r$ is sufficiently large.

Hence, for $r>0$ large enough, we can apply \eqref{2.17} and \eqref{2.18} to deduce that
\be\label{2.21}
\frac{a'(r)}{r}>\int_r^\infty\frac{e_1^2}{s}\left(\frac{1}{e_1}-a(s)\right)\dd s.
\ee
Therefore, $a'(r)>0$ for $r>0$ sufficiently large. Using this and \eqref{2.5}, we see that $a(r)$ approaches its limiting value $a_\infty$ as $r\to\infty$. Moreover, $f(\infty)=1$ and the convergence of the integral $\int_0^\infty\frac{(1-e_1 a(r))^2}{r}f(r)\dd r$ imply $a_\infty=1/e_1$.

In the following, we will show that $\lim_{r\to 0}b(r)=C_1$. We claim
\be\label{2.22}
\liminf_{r\to 0}r|b'(r)|=0.
\ee
In fact, if \eqref{2.22} fails to exist, there are constants $\vep_1>0$ and $\delta>0$ so that
\be\label{2.23}
r|b'(r)|\geq\vep_1,~~0<r<\delta.
\ee
Thus, we have from \eqref{2.23} that
\be\label{2.24}
\int_{0}^{\delta}r(b'(r))^2\dd r\geq\int_{0}^{\delta}\frac{\vep_1^2}{r}\dd r=\infty.
\ee
This result is in contradiction to the convergence of the integral $\int_0^\infty\frac{\omega^2 r(b'(r))^2}{2}\dd r$.

Integrating \eqref{1.10} and using \eqref{2.22}, we obtain
\be\label{2.25}
r b'(r)=\int_0^r s\left(2b(s)e_1^2f^2(s)+2e_2^2\left(b(s)-\frac{1}{e_2}\right)g^2(s)\right)\dd s.
\ee
Note that $f(r)=O(r^{N(1-\vep)}), g(r)=O(r^{M(1-\vep)})$ as $r\to 0$ (these two estimates will be show in the next section), we see that $b'(r)\geq0$ when $r>0$ small. Combine this fact with $b(r)$ is bounded for all $r\geq 0$, we see that there exists a number $C_1$ so that $\lim_{r\to 0}b(r)=C_1$.
\end{proof}

In conclusion, we have shown that $(a,b,f,g)$ is a least energy solution of \eqref{1.9}--\eqref{1.12} subject to the boundary conditions \eqref{1.17}--\eqref{1.18}.

\section{Asymptotic estimates}\label{s3}
\setcounter{equation}{0}

In this section we will present some properties of the solution of \eqref{1.9}--\eqref{1.12} obtained in the last section. First, we have

\begin{lemma}\label{le.3}
For the twisted vortices solution $(a,b,f,g)$ of \eqref{1.9}--\eqref{1.12} obtained in the last section, there hold the decay estimates
\bea
a(r)&=&O(r^{2(1-\vep)}),~~b(r)=C_1+O(r^{2N(1-\vep)+2}),\nn\\
f(r)&=&O(r^{N(1-\vep)}),~~g(r)=O(r^{M(1-\vep)}),\label{3.1}
\eea
for $r\to 0$, where $0<\vep<1$ is an arbitrary number.
\end{lemma}
\begin{proof}
Take the comparison function $A_\vep(r)=Cr^{2(1-\vep)}$, where $C>0$ is a constant to be chosen later. From the equation \eqref{1.9} and the conditions $f(0)=g(0)=0$, we see that there is an $r_\vep>0$ small so that
\be\label{3.2}
(a-A_\vep)''-\frac{(a-A_\vep)'}{r}\geq (2f^2e_1^2+2g^2e_2^2)(a-A_\vep),~~r\in(0, r_\vep).
\ee
For such fixed $r_\vep$, let the constant $C>0$ be large enough to make $a(r_\vep)-A_\vep(r_\vep)\leq 0$. In view of this and the boundary condition $(a-A_\vep)(r)\to 0$ as $r\to 0$, we obtain by applying the maximum principle in \eqref{3.2} the result $a(r)\leq A_\vep(r)$ for $r\in(0, r_\vep)$, which establish the bound
\be\label{3.3}
0<a(r)\leq C r^{2(1-\vep)},~~r\in(0, r_\vep),
\ee
resulting in the estimate for $a$ in \eqref{3.1}.

Next, set $F_\vep(r)=Cr^{N(1-\vep)}$, then using $f(0)=g(0)=0$ and \eqref{1.11}, we have
\be\label{3.5}
(f-F_\vep)''+\frac{(f-F_\vep)'}{r}\geq\frac{(1-e_1a)^2N^2}{r^2}(f-F_\vep),~~r\in(0, r_\vep),
\ee
where $r_\vep$ is chosen to be sufficiently small. Taking $C$ large enough such that $f(r_\vep)-F_\vep(r_\vep)\leq 0$. By virtue of this, $f(0)=F_\vep(0)=0$, the differential equality \eqref{3.5}, and the maximum principle, we have $f(r)\leq F_\vep(r)$ for all $r\in(0, r_\vep)$. That is
\be\label{3.6}
0<f(r)\leq C r^{N(1-\vep)},~~r\in(0, r_\vep),
\ee
which gives rise to the asymptotic estimate for $f(r)$ near $r=0$ stated in \eqref{3.1}.

Now consider the estimate for $g(r)$ at $r=0$. Take $G_\vep(r)=Cr^{M(1-\vep)}$, then \eqref{1.12} and $f(0)=g(0)=0$ lead to
\be\label{3.7}
(g-G_\vep)''+\frac{(g-G_\vep)'}{r}\geq\frac{(e_2Na-M)^2}{r^2}(g-G_\vep),
\ee
where, again, $r_\vep>0$ is sufficiently small. Choose $C>0$ large enough to make $g(r_\vep)-G_\vep(r_\vep)\leq 0$. Using this and that $g-G_\vep$ vanishes at origin in \eqref{3.7}, we obtain $g(r)\leq G_\vep(r)$ for $r\in(0, r_\vep)$ with an application of the maximum principle. This results in the bound
\be\label{3.8}
0<g(r)\leq C r^{M(1-\vep)},~~r\in(0, r_\vep).
\ee
So the decay estimate for $g(r)$ near $r=0$ is verified.

Finally, let $B(r)=b(r)-C_1$, we can rewrite the equation \eqref{1.10} as
\be\label{3.9}
(r B')'=2r(B+C_1)(e_1^2f^2+e_2^2g^2)-2e_2rg^2.
\ee
Since $B(r)\to 0$ as $r\to 0$, we get
\be\label{3.10}
(r B')'\sim 2C_1e_1^2rf^2+(2C_1e_2^2-2e_2)rg^2,~~r\to 0.
\ee
Then \eqref{3.10} and the estimates for $f(r)$ and $g(r)$ near $r=0$ give us
\be\label{3.11}
(r B'(r))'=O(r^{2N(1-\vep)+1}),~~r\to 0.
\ee

Integrating \eqref{3.11} on $(0, r)$, we obtain $r B'(r)=O(r^{2N(1-\vep)+2})$ as $r\to 0$, or
\be\label{3.12}
B'(r)=O(r^{2N(1-\vep)+1}),~~r\to 0.
\ee
Again, integrating \eqref{3.12} on $(0, r)$, we have
\be\label{3.13}
B(r)=O(r^{2N(1-\vep)+2}),~~r\to 0.
\ee
\end{proof}
We now study the decay estimates of the solution in the limit $r\to\infty$.

\begin{lemma}\label{le.4}
The solution quartet $(a,b,f,g)$ of \eqref{1.9}--\eqref{1.12} satisfies the following asymptotic estimates
\bea
a(r)&=&1/e_1+O(\e^{-\sqrt{2}e_1(1-\vep)r}),~~b(r)=O(\e^{-\sqrt{2}e_1(1-\vep)r}),\nn\\
f(r)&=&O(\e^{-\sqrt{2}\beta_1(1-\vep)r}),~~g(r)=O(\e^{-\sqrt{\omega^2-\alpha+\beta'}(1-\vep)r}),\label{3.14}
\eea
for $r\to \infty$, where $\vep$ is an arbitrary number lying in the interval $(0, 1)$.
\end{lemma}
\begin{proof}
Introduce the comparison function
\be\label{3.15}
a_\vep(r)=C\e^{-\sqrt{2}e_1(1-\vep)r},~~r>0,~~\vep\in(0,1).
\ee
In view of $f(\infty)=1$ and $g(\infty)=0$, for any fixed $\vep>0$ and $C\geq 1$ given in \eqref{3.15}, there is some large $r_\vep>1$ such that
\bea\label{3.19}
\left(a-\frac{1}{e_1}+a_\vep\right)''-\frac{(a-\frac{1}{e_1}+a_\vep)'}{r}\leq(2f^2e_1^2+2g^2e_2^2)\left(a-\frac{1}{e_1}+a_\vep\right),~~r\geq r_\vep.
\eea
For such fixed $r_\vep$, we can choose constant $C>0$ in \eqref{3.15} large enough to make $a(r_\vep)-\frac{1}{e_1}+a_\vep(r_\vep)\geq 0$. Using the boundary condition $a(\infty)-\frac{1}{e_1}+a_\vep(\infty)=0$ and applying the maximum principle to the differential inequality \eqref{3.19}, we have $a(r)-\frac{1}{e_1}+a_\vep(r)\geq 0$ for $r\geq r_\vep$. Thus, we see that the solution $a$ of \eqref{1.9} verifies $a(r)=1/e_1+O(\e^{-\sqrt{2}e_1(1-\vep)r})$ as $r\to\infty$.

Next, we consider the asymptotic of $b(r)$ at infinity. Linearize the equation \eqref{1.10} at $r=\infty$ gives us $b''=2e_1^2b$ whose characteristic roots are $\sqrt{2}e_1$ and $-\sqrt{2}e_1$. Thus, we may use the method in the proof of the estimate for $a(r)$ at $\infty$ to get the estimate for $b(r)$ claimed in \eqref{3.14}.

For $f$ we rewrite the equation \eqref{1.11} as
\be\label{3.20}
(f-1)''+\frac{(f-1)'}{r}=\beta_1f(f+1)(f-1)+f\left(\frac{(1-e_1a)^2N^2}{r^2}+e_1^2\omega^2b^2+\beta'g^2\right).
\ee
The structure of \eqref{3.20} leads us to its linearized form around $r=\infty$,
\be\label{3.21}
F''=2\beta_1 F,
\ee
whose solution is $F(r)=C\e^{-\sqrt{2}\beta_1r}$ near $r=0$. This suggests that we may choose the comparison function
\be\label{3.22}
f_\vep(r)=C\e^{-\sqrt{2}\beta_1(1-\vep)r},~~r>0,~~\vep\in(0, 1),
\ee
which produces that
\be\label{3.23}
f_\vep''+\frac{f_\vep'}{r}=2\beta_1^2(1-\vep)^2 f_\vep-\frac{\sqrt{2}\beta_1(1-\vep)}{r}f_\vep.
\ee
Combining \eqref{3.20} with \eqref{3.23}, when $r_\vep>0$ is sufficiently large, we have
\be\label{3.24}
(f-1+f_\vep)''+\frac{(f-1+f_\vep)'}{r}\leq\beta_1f(f+1)(f-1+f_\vep),~~r\geq r_\vep,
\ee
where we used $b(\infty)=0$, $f(\infty)=1$ and $g(\infty)=0$. Choose $C>0$ in \eqref{3.22} large enough such that $(f-1+F_\vep)(r_\vep)\geq 0$. Note the boundary condition $(f-1+F_\vep)(\infty)=0$ and applying the maximum principle to the inequality \eqref{3.24}, we have $f(r)-1+F_\vep(r)\geq 0$ for $r\geq r_\vep$. So the decay estimate for $f$ near infinity stated in \eqref{3.14} is established.

Finally, we rewrite the equation \eqref{1.12} as
\be\label{3.25}
g''+\frac{g'}{r}=g\left(\frac{(e_2Na-M)^2}{r^2}+\omega^2(1-e_2b)^2+\beta_2g^2-\alpha+\beta'f^2\right).
\ee
Since $b(\infty)=0$ and $f(\infty)=1$, and the characteristic roots of $g''=(\omega^2-\alpha+\beta')g$ are $\sqrt{\omega^2-\alpha+\beta'}$, we see that the solution $g$ of \eqref{3.25} satisfies $g(r)=O(\e^{-\sqrt{\omega^2-\alpha+\beta'}(1-\vep)r})$ as $r\to\infty$.
\end{proof}

\medskip



\end{document}